\documentclass[cclayout]{cc}

\usepackage{enumerate}
\usepackage{boxedminipage}
\usepackage{eucal}
\usepackage{ccalgo}
\usepackage{rotating}
\usepackage{ctable}

\contact{S070006@ntu.edu.sg}
\received{8 February 2010}
\title{Query-Efficient Locally Decodable Codes of Subexponential Length}
\titlehead{Query-Efficient Locally Decodable Codes}
\author{
    Yeow Meng Chee \\
    Division of Mathematical Sciences \\
    School of Physical \& Mathematical \\ Sciences \\
    Nanyang Technological University \\
    Singapore 637371 \\
    \email{ymchee@ntu.edu.sg} \\
    \and
    Tao Feng \\
    Department of Mathematical Sciences \\
    University of Delaware \\
    Newark, DE 19716, USA \\
    \email{feng@math.udel.edu} \\
    \and
    San Ling \\
    Division of Mathematical Sciences \\
    School of Physical \& Mathematical \\ Sciences \\
    Nanyang Technological University \\
    Singapore 637371 \\
    \email{lingsan@ntu.edu.sg} \\
    \and
    Huaxiong Wang \\
    Division of Mathematical Sciences \\
    School of Physical \& Mathematical \\ Sciences \\
    Nanyang Technological University \\
    Singapore 637371 \\
    \email{hxwang@ntu.edu.sg} \\
    \and
    Liang Feng Zhang \\
    Division of Mathematical Sciences \\
    School of Physical \& Mathematical \\ Sciences \\
    Nanyang Technological University \\
    Singapore 637371 \\
    \email{liangf.zhang@gmail.com}
}
\authorhead{Chee, Feng, Ling, Wang \& Zhang}

\begin{abstract}
A $k$-query locally decodable code (LDC)
$\textbf{C}:\Sigma^{n}\rightarrow \Gamma^{N}$ encodes each message $x$ into
a codeword $\textbf{C}(x)$ such that each symbol of $x$ can be probabilistically
recovered by querying only $k$ coordinates of $\textbf{C}(x)$, even after a
constant fraction of the coordinates have been corrupted.
Yekhanin (2008)
constructed a $3$-query LDC of subexponential length,
$N=\exp(\exp (O(\log n/\log\log n)))$, under the assumption that there are
infinitely many Mersenne primes. Efremenko (2009) constructed a $3$-query LDC
of length $N_{2}=\exp(\exp (O(\sqrt{\log n\log\log n})))$ with no assumption, and a
$2^r$-query LDC  of length $N_{r}=\exp(\exp(O(\sqrt[r]{\log n(\log \log n)^{r-1}})))$,
for every integer $r\geq 2$. Itoh and Suzuki (2010) gave a composition method in
Efremenko's framework and constructed a $3 \cdot 2^{r-2}$-query LDC of length
$N_{r}$, for every integer $r\geq 4$, which improved the query complexity of
Efremenko's LDC of the same length by a factor of $3/4$.
The main ingredient of
Efremenko's construction is the Grolmusz construction for super-polynomial
size set-systems with restricted intersections, over $\mathbb{Z}_m$, where
$m$ possesses a certain ``good'' algebraic property (related to the
``algebraic niceness'' property of Yekhanin (2008)).
Efremenko constructed a 3-query LDC  based on $m=511$
and  left as an open problem to find other
numbers that offer the same property for LDC constructions. 

In this paper, we develop the algebraic theory behind the constructions of
Yekhanin (2008) and Efremenko (2009), in an attempt to understand
the ``algebraic niceness'' phenomenon in $\mathbb{Z}_m$.
We show that every integer
$m = pq = 2^t -1$, where $p$, $q$ and $t$ are prime, possesses the same
good algebraic property as $m=511$ that allows savings in query complexity.
We identify 50 numbers of this form by
computer search, which together with 511, are then applied to gain 
improvements on query complexity via Itoh and Suzuki's composition method.
More precisely,
we  construct a $3^{\lceil r/2\rceil}$-query LDC for every positive integer
$r<104$ and a $\left\lfloor (3/4)^{51}\cdot 2^{r}\right\rfloor$-query LDC
for every integer $r\geq 104$,
both of length $N_{r}$, improving the $2^r$ queries used by Efremenko (2009) and
$3\cdot 2^{r-2}$ queries used by Itoh and Suzuki (2010). 

We also obtain new efficient private information retrieval (PIR)
schemes from the new query-efficient LDCs.
\end{abstract}


\begin{subject}
20C05, 94B60
\end{subject}

\begin{document}

\section{Introduction} \label{Sec:introduction}

A classical error-correcting code $\textbf{C}:\Sigma^{n}\rightarrow
\Gamma^{N}$ allows one to encode a message $x$
 into a codeword $\textbf{C}(x)$ such that $x$
can be   recovered even if $\textbf{C}(x)$ gets corrupted in a
number of coordinates. However, to recover even a small portion of
the message $x$, one has to consider all or most of the coordinates
of the received (possibly corrupted) codeword.
\citet{KatzTrevisan:2000}
considered error-correcting codes where each symbol of the message
can be probabilistically recovered by looking at a limited
number of coordinates of a corrupted encoding.
Such codes are known as {\it locally decodable codes}
(LDCs). Informally, a $(k,\delta,\epsilon)$-LDC $\textbf{C}:\Sigma^{n}\rightarrow
\Gamma^{N}$ encodes a message $x$ into a codeword
$\textbf{C}(x)$ such that each symbol $x_{i}$ of the
message can be recovered with probability at least $1-\epsilon$, by a probabilistic
decoding algorithm that makes at most $k$ queries,
 even if the codeword is corrupted in up to $\delta N$ locations.
LDCs have many applications in cryptography and complexity theory
(see, for example, \citet{Gasarch:2004,Trevisan:2004}),
and have attracted a considerable amount of attention
\citep{Deshpandeetal:2002,Obata:2002,KerenidisWolf:2004,DvirShpilka:2005,WehnerWolf:2005,Goldreichetal:2006,ShiowattanaLokam:2006, Raghavendra:2007,Woodruff:2007, KedlayaYekhanin:2008,Yekhanin:2008,Efremenko:2009,Gopalan:2009,ItohSuzuki:2010} since their formal introduction by
\citet{KatzTrevisan:2000}.

For constant $\delta$ and $\epsilon$, the efficiency of a
$(k,\delta,\epsilon)$-LDC $\textbf{C}:\Sigma^{n}\rightarrow
\Gamma^{N}$ is measured by its \emph{length} $N$ and
\emph{query complexity} $k$. Ideally, we want both $N$ and $k$ to be
as small as possible.
\citet{KatzTrevisan:2000} proved that there do not exist
families of 1-query LDCs.  \citet{Goldreichetal:2006} obtained an
exponential lower bound of $\exp(\Omega(n))$ on the length of 2-query \emph{linear}
LDCs. \citet{KerenidisWolf:2004} showed that the optimal length of \emph{any} 2-query LDCs
is $\exp(O(n))$  via a quantum argument.
For a $k$-query ($k\geq 3$) LDC,
\citet{Woodruff:2007} obtained  a superlinear lower bound of
$\Omega(n^{(k+1)/(k-1)}/\log n)$ on its length. Other lower bounds have been obtained by
\citet{Deshpandeetal:2002}, \citet{Obata:2002}, \citet{DvirShpilka:2005}, \citet{WehnerWolf:2005},  and \citet{ShiowattanaLokam:2006}.

It has been conjectured for a long time that the length $N$ of any
constant-query LDC should have an exponential dependence on its
message length $n$. This conjecture was disproved by
\citet{Yekhanin:2008}, who constructed a 3-query LDC of length $\exp(\exp
(O(\log n/\log\log n)))$ under the assumption that there are
infinitely many {\em Mersenne primes} (primes of the form
$M_{t}=2^{t}-1$, where $t$ is prime). Subsequently, Yekhanin's
construction was nicely reformulated by \citet{Raghavendra:2007}
using group homomorphism. Inspired by this,
\citet{Efremenko:2009} generalized Yekhanin's construction and established a
framework for constructing LDCs in which the above assumption on
Mersenne primes is no longer necessary. \citet{Efremenko:2009} constructed  a
$k_{r}$-query ($k_{r}\leq 2^{r}$) LDC of length
$N_{r}=\exp(\exp(O(\sqrt[r]{\log n(\log \log n)^{r-1}})))$ for every
integer $r\geq 2$, and in particular, a 3-query ($k_{2}=3$) LDC of length
 $N_{2}=\exp(\exp (O(\sqrt{\log n\log\log n})))$ for $r=2$.
The main ingredient of Efremenko's construction is a construction of \citet{Grolmusz:2000}
for super-polynomial size set-systems with restricted intersections. Each of these set-systems
is over a
certain composite number, which has significant impact on the query
complexity (the value of
$k_{r}$) of the resulting LDC.
\citet{Efremenko:2009} showed that the composite number  511
can result in a 3-query LDC of length $N_{2}$ and left as an open problem to find
other suitable composite numbers.

Recently, \citet{ItohSuzuki:2010} developed a composition method
in Efremenko's framework. This method allows one to
compose, in an appropriate way, Efremenko's $k_{r}$-query
($k_{r}\leq 2^{r}$) LDC of length $N_{r}$ and $k_{l}$-query
($k_{l}\leq 2^{l}$) LDC of length $N_{l}$ to obtain a $k$-query LDC
 of length
$N_{r+l}$ such that  $k\leq k_{r}k_{l}$. For every integer $r\geq 4$, taking
Efremenko's 3-query LDC
 and  $k_{r-2}$-query LDC
 as
building blocks, the composition method yields a $k$-query LDC
 of length $N_{r}$ in which  $k\leq 3\cdot 2^{r-2}$,
improving the query
complexity of Efremenko's LDC of the same length by a factor of $3/4$.
We stress that this improvement
is due to the first
building block, that is, the 3-query LDC. Hence, it is of great
interest to obtain as many such 3-query LDCs as possible, or equivalently,
as many new composite numbers as possible which can result in
3-query LDCs of length $N_{2}$ in Efremenko's construction.

\subsection{Our Results}
In this paper we study the algebraic properties of {\em good}
composite numbers which yield 3-query LDCs in
Efremenko's construction. We give a  characterization of such
composite numbers and show that every Mersenne number which is a
product of two primes is good. Consequently, we obtain a number of
good composite numbers. These new good numbers,
together with 511, are then applied to achieve improvements on
the query complexity in Efremenko's framework.

Let $\mathbb{M}_{2}$ be the set of composite numbers, each of which
is the product of two distinct odd primes and good (i.e., can yield a 3-query
LDC of length $N_{2}$ in Efremenko's construction).
We characterize numbers in
$\mathbb{M}_{2}$, and show that the subset of
{\em Mersenne numbers} (numbers of the form $M_{t}=2^{t}-1$, where $t$ is prime)
\begin{equation*}
\mathbb{M}_{2,{\rm Mersenne}} = \{m : \text{$m=2^t-1=pq$, where
$p$, $q$ and $t$ are primes} \}
\end{equation*}
is contained in $\mathbb{M}_{2}$. Note that the number $511 = 2^9 -1
= 7 \times 73$, suggested by \citet{Efremenko:2009}, is in
$\mathbb{M}_{2}$ but not in $\mathbb{M}_{2,\rm Mersenne}$. On the
other hand, the number $15 = 3\times 5$, the smallest possible
candidate for $\mathbb{M}_{2}$, is not in $\mathbb{M}_{2}$, checked
via exhaustive search by \citet{ItohSuzuki:2010}.
We identify 50 numbers in $\mathbb{M}_{2,{\rm Mersenne}}$ and hence
50 new
 numbers in $\mathbb{M}_2$, which answers open problems
raised by \citet{Efremenko:2009} and \citet{ItohSuzuki:2010}.
Furthermore,  we show that:
\begin{enumerate}[(a)]
\item For every integer $r$, $1\leq r\leq 103$, there is a
$k$-query linear LDC of length $N_{r}$ for which
\begin{equation*}
k \leq \begin{cases}
(\sqrt{3})^r,&\text{if $r$ is even} \\
8\cdot (\sqrt{3})^{r-3},&\text{if $r$ is odd.}
\end{cases}
\end{equation*}
\item
For every integer $r\geq 104$, there is a $k$-query linear LDC of length
$N_{r}$ for which $k\leq (3/4)^{51}\cdot 2^{r}$.
\item
If $|\mathbb{M}_{2,{\rm Mersenne}}|=\infty$, then for every integer
$r\geq 1$, there is a $k$-query linear LDC of length $N_{r}$ for which
$k$ is the same as that in (a).
\end{enumerate}

The notion of LDCs is closely related to the notion of
information-theoretic private information retrieval (PIR) schemes.
It is well known that LDCs with perfectly smooth decoders imply PIR
schemes, and there is a generic transformation from LDCs to PIR
schemes \citep{KatzTrevisan:2000}. As with the LDCs of
\citet{Efremenko:2009} and \citet{ItohSuzuki:2010}, the query-efficient LDCs
obtained in this paper also have perfectly smooth
decoders\footnote[1]{Note that the decoders for the LDCs
of \citet{Yekhanin:2008} are not smooth.}.
This in
turn gives new PIR schemes with smaller communication complexity. For
instance, the LDCs from (a) above imply PIR schemes with
communication complexity $\exp(O(\sqrt[r]{\log n(\log\log n)^{r-1}}))$ for
$3^{ r/2}$ servers. Compared with the best known PIR schemes
of \citet{ItohSuzuki:2010} with the same communication complexity
for $3\cdot2^{r-2}$ servers,
where $r<104$ is even, our new schemes require fewer
servers.

We are able to identify only
50 numbers in $\mathbb{M}_{2,{\rm Mersenne}}$ by computer search with
the largest one being $M_{7331}=2^{7331}-1$.
We believe that the search for more numbers in $\mathbb{M}_{2,{\rm
Mersenne}}$ is of independent interest. In particular,  it is an interesting open problem to
determine how many numbers $\mathbb{M}_{2,\rm Mersenne}$ contains.
Compared with Mersenne primes, it seems reasonable to conjecture that
$|\mathbb{M}_{2, \rm Mersenne}|=\infty$.

\subsection{Organization}
This paper is organized as follows. In
\ref{Sec:preliminaries}, we review Efremenko's framework and the
composition method of \citet{ItohSuzuki:2010}. In
\ref{Sec:characterization},
we prove that all Mersenne numbers which are
products of two primes belong to $\mathbb{M}_{2}$ and
introduce the family $\mathbb{M}_{2,{\rm Mersenne}}$. We
also characterize the numbers in $\mathbb{M}_{2}$ and discuss how to
prove that a given number is not in $\mathbb{M}_{2}$.
 In \ref{Sec:LdcAndPir}, we
obtain new query-efficient LDCs using the
 family $\mathbb{M}_{2,{\rm Mersenne}}$. This also gives  new efficient
 PIR schemes with fewer servers. We conclude the paper in \ref{Sec:conclusion}.

 \section{Preliminaries}
 \label{Sec:preliminaries}

We briefly review
Efremenko's framework \citep{Efremenko:2009} and the composition method of
\citet{ItohSuzuki:2010}.

Let $m$ and $h$ be positive integers. The ring
$\mathbb{Z}/m\mathbb{Z}$ is denoted $\mathbb{Z}_m$. The set
$\{1,2,\ldots,m\}$ is denoted $[m]$. The {\em $\bmod$ $m$ inner
product} of two vectors $x = (x_1, \ldots , x_h) ,y =(y_1, \ldots ,
y_h) \in \mathbb{Z}_{m}^{h}$ is defined to be $\langle
x,y\rangle_{m}\equiv\sum_{i=1}^{h}x_{i} y_{i}\bmod m$. The {\em
Hamming distance} between $x$ and $y$ is denoted $d_H(x,y)$.

\begin{definition}[Locally Decodable Code]
Let $k$, $n$ and $N$ be positive integers, and $0 < \delta,\epsilon< 1$.
\label{Def:LDC} A  code $\emph{\textbf{C}}:\Sigma^{n}\rightarrow
\Gamma^{N}$ is said to be $(k,\delta,\epsilon)$-{\em locally
decodable} if there is a probabilistic decoding algorithm
$\mathcal{D}$ such that:
\begin{enumerate}
  \item For every $x\in \Sigma^{n}$, $i\in[n]$, and $y\in
  \Gamma^{N}$ such that $d_H(y,\emph{\textbf{C}}(x))\leq \delta
  N$, we have $\Pr[\mathcal{D}^{y}(i)=x_{i}]\geq 1-\epsilon$, where $\mathcal{D}^{y}$
means that
 $\mathcal{D}$ makes oracle access to $y$, and
  the probability is taken over the internal coin tosses of
$\mathcal{D}$.
  \item In every invocation, $\mathcal{D}$ makes at most $k$
  queries to  $y$.
\end{enumerate}
\end{definition}

The algorithm $\mathcal{D}$ is called a
$(k,\delta,\epsilon)$-{\em local decoding algorithm} for
$\textbf{C}$. Parameters $k$ and $N$ are called the {\em query complexity}
and {\em length} of $\textbf{C}$, respectively. The
{\em alphabets} $\Sigma$ and $\Gamma$ are often taken to be a finite field
$\mathbb{F}_{q}$, where $q$ is a prime power.
 A $k$-query LDC
$\textbf{C}:\mathbb{F}_{q}^{n}\rightarrow \mathbb{F}_{q}^{N}$ is
\textit{linear} if it is a linear transformation,  and
\textit{nonadaptive} if in every invocation, $\mathcal{D}$ makes all
 queries simultaneously. All the LDCs in this paper are linear
and nonadaptive.

\subsection{Efremenko's Framework}
\label{subSec:EfrConstruction}

Efremenko's  framework \citep{Efremenko:2009} for constructing LDCs is essentially a
generalization of the work of \citet{Yekhanin:2008}. Let $m=p_{1}p_{2}\ldots
p_{r}$ be a product of $r \geq 2$ distinct odd primes
$p_{1},p_{2},\ldots,p_{r}$. Let $S\subseteq
\mathbb{Z}_{m}\setminus\{0\}$ and $h$ be a positive integer.
Let $t$ be the multiplicative order of $2 \in\mathbb{Z}_{m}^{*}$, and
let $\gamma_{m}\in \mathbb{F}_{2^{t}}^{*}$ be a primitive $m$-th
root of unity. The  building blocks of Efremenko's framework for constructing LDCs
include both an $S$-{\em matching family} and an $S$-{\em decoding polynomial}, which are
defined as follows:

\begin{definition}[$S$-Matching Family]
\label{Def:S_matching}
For $S\subseteq\mathbb{Z}_{m}\setminus\{0\}$, a family of vectors
$\{u_{i}\}_{i=1}^{n}\subseteq \mathbb{Z}_{m}^{h}$ is called an
{\em $S$-matching family} if:
\begin{enumerate}
  \item $\langle u_{i},u_{i}\rangle_{m}=0$,  for $i\in[n]$; and
  \item $\langle u_{i},u_{j}\rangle_{m}\in S$, for distinct $i,j\in [n]$.
\end{enumerate}
\end{definition}

\begin{definition}[$S$-Decoding Polynomial]
\label{Def:S_decoding}
For $S\subseteq\mathbb{Z}_{m}\setminus\{0\}$, a
polynomial $P(X)\in\mathbb{F}_{2^{t}}[X]$ is called
an {\em $S$-decoding polynomial} if:
\begin{enumerate}
  \item $P(\gamma_{m}^{s})=0$, for $s\in S$; and
  \item $P(\gamma_m^{0})=P(1)=1$.
\end{enumerate}
\end{definition}

For any subset $S\subseteq \mathbb{Z}_{m}\setminus \{0\}$, an $S$-matching family and the corresponding $S$-decoding polynomial yield a linear LDC immediately.

\begin{theorem}[\citet{Efremenko:2009}]
\label{Thm:Efr1}
Let
$\{u_{i}\}_{i=1}^{n}\subseteq \mathbb{Z}_{m}^{h}$ be an $S$-matching
family and  $P(X)=a_{0}+a_{1}X^{b_{1}}+\ldots+a_{k-1}X^{b_{k-1}}\in
\mathbb{F}_{2^{t}}[X]$ be an $S$-decoding polynomial with $k$
monomials. Then there is a $k$-query linear LDC $\emph{\textbf{C}}:
\mathbb{F}_{2^{t}}^{n}\rightarrow \mathbb{F}_{2^{t}}^{m^{h}}$ with
encoding and decoding algorithms as in Fig. \bare\ref{Alg:Efr}.
\end{theorem}

\begin{figure}[t]
\label{Alg:Efr}
\underline{\em Encoding} \\
Let $e_{j}\in \mathbb{F}_{2^{t}}^{n}$ denote the $j$-th unit vector for $j\in[n]$.
The coordinates of a codeword $\textbf{C}(x)$ are
indexed by vectors in $\mathbb{Z}_{m}^{h}$, where $x\in \mathbb{F}_{2^{t}}^{n}$.
The encoding algorithm works as follows:
\begin{enumerate}
\item for $j\in [n]$ and $v\in \mathbb{Z}_{m}^{h}$,
$\textbf{C}(e_{j})_{v}=\gamma_{m}^{\langle u_{j},v\rangle_{m}}$;
\item for $x=(x_{1},\ldots,x_{n})\in \mathbb{F}_{2^{t}}^{n}$, we have
$\textbf{C}(x)=\sum_{j=1}^{n}x_{j}\cdot \textbf{C}(e_{j})$. \\
\end{enumerate}

\underline{\em Decoding} \\
To recover $x_{i}$ from
a possibly corrupted codeword $y\in \mathbb{F}_{2^{t}}^{m^{h}}$ of any message $x$,we
\begin{enumerate}
\item choose a vector $v\in \mathbb{Z}_{m}^{h}$ uniformly and query the coordinates
$y_{v},y_{v+b_{1}u_{i}},\ldots, y_{v+b_{k-1}u_{i}}$;
\item output $\gamma_{m}^{-\langle u_{i},v\rangle_{m}}\cdot (a_{0}\cdot y_{v}+a_{1}\cdot y_{v+b_{1}u_{i}}+
  \ldots+a_{k-1}\cdot y_{v+b_{k-1}u_{i}})$.
\end{enumerate}
\hrule
\caption{Efremenko's Framework for Constructing LDCs}
\end{figure}

\ref{Thm:Efr1} shows that for any $S\subseteq
\mathbb{Z}_{m}\setminus\{0\}$, an $S$-matching family of size $n$
and an $S$-decoding polynomial with $k$ monomials yield a $k$-query
LDC which encodes each message of length $n$ into a codeword of
length $m^{h}$. Once  $m$ and $h$ are fixed, the  length $N$
is inversely proportional to  $n$. Hence, ideally, $n$ should be large and $k$ small.
To have a large $S$-matching family, the set $S$ is usually taken to be
$S_{m}$, the {\em canonical set} of $m$, which is defined as follows:

\begin{definition}[Canonical Set]
\label{Def:Canonical_set} Let $m=p_{1}p_{2}\ldots p_{r}$ be the
product of $r \geq 2$ distinct odd primes
$p_{1},p_{2},\ldots,p_{r}$. The {\em canonical set} of $m$ is
defined to be
\begin{equation*}
S_{m}=\left\{s_{\sigma}\in \mathbb{Z}_{m} : \text{ $\sigma\in
\{0,1\}^{r}\setminus\{{\bf 0}\}$ and $s_{\sigma}\equiv
\sigma_{i}\bmod{p_{i}}$, for $i\in[r]$}\right\}.
\end{equation*}
\end{definition}

For every integer $r \geq 2$, \citet{Efremenko:2009} proved that
there exist an $S_{m}$-matching family of superpolynomial size and
an $S_{m}$-decoding polynomial with at most $2^{r}$ monomials.

\begin{proposition}[\citep{Efremenko:2009}]
\label{Lem:Efr}
Let
$m=p_{1}p_{2}\ldots p_{r}$ be the product of $r \geq 2$ distinct odd
primes $p_{1},p_{2},\ldots,p_{r}$.
\begin{enumerate}
\item There is a positive constant $c$, depending only on $m$,
such that for every integer $h>0$,
there is an $S_{m}$-matching family $\{u_{i}\}_{i=1}^{n}\subseteq
\mathbb{Z}_{m}^{h}$ of size $n\geq \exp\left(c(\log h)^{r}/(\log\log
h)^{r-1}\right)$.
\item There is an $S_{m}$-decoding
polynomial with at most $2^{r}$ monomials.
\end{enumerate}
\end{proposition}

Efremenko's linear LDCs of subexponential length now immediately follow from
\ref{Thm:Efr1} and \ref{Lem:Efr}.

\begin{theorem}[\citep{Efremenko:2009}]
\label{Thm:Efr2}
For every integer $r \geq 2$, there is a linear
$(k_{r},\delta,k_{r}\delta)$-LDC of length
$N_{r}=\exp(\exp(O(\sqrt[r]{\log n(\log \log n)^{r-1}})))$ for which
$k_{r}\leq 2^r$. In particular, when $r=2$, there is a linear
$(3,\delta,3\delta)$-LDC  of length $N_{2}=\exp(\exp (O(\sqrt{\log n\log\log n})))$ .
\end{theorem}

\subsection{The Composition Method}

For every integer $r \geq 2$, there is a $k_{r}$-query linear LDC of
subexponential length $N_{r}$ by \ref{Thm:Efr2}, but its
query complexity $k_{r}$ is only upper bounded by $2^{r}$. It is
attractive to improve the query complexity.
This is the motivation for Itoh and Suzuki's  composition method.

Let $m_{1}=p_{1}p_{2}\ldots p_{r}$ be the product of $r$ distinct odd
primes $p_{1},p_{2}\ldots,p_{r}$ and $ m_{2}=q_{1}q_{2}\ldots
q_{l}$ the product of $l$ distinct odd primes
$q_{1},q_{2}\ldots,q_{l}$, where $r, l\geq 2$. Suppose
$\gcd(m_{1},m_{2})=1$. Let $m=m_{1}m_{2}$, and $t_{1}$, $t_{2}$, and $t$ be
the multiplicative orders of 2 in
$\mathbb{Z}_{m_{1}}^{*}$, $\mathbb{Z}_{m_{2}}^{*}$, and
$\mathbb{Z}_{m}^{*}$, respectively. By \ref{Thm:Efr1} and
\ref{Thm:Efr2}, there are linear LDCs
$\textbf{C}_{r}:\mathbb{F}_{2^{t_{1}}}^{n}\rightarrow
\mathbb{F}_{2^{t_{1}}}^{N_{r}}$,
$\textbf{C}_{l}:\mathbb{F}_{2^{t_{2}}}^{n}\rightarrow
\mathbb{F}_{2^{t_{2}}}^{N_{l}}$ and
$\textbf{C}_{r+l}:\mathbb{F}_{2^{t}}^{n}\rightarrow
\mathbb{F}_{2^{t}}^{N_{r+l}}$ of query complexities $k_{r}\leq
2^{r}$, $k_{l}\leq 2^{l}$, and $k_{r+l}\leq 2^{r+l}$, respectively.
Let $P_{1}(X)\in \mathbb{F}_{2^{t_{1}}}[X]$ and $P_{2}(X)\in
\mathbb{F}_{2^{t_{2}}}[X]$ be the $S_{m_{1}}$-decoding polynomial
for $\textbf{C}_{r}$ and $S_{m_{2}}$-decoding polynomial for
$\textbf{C}_{l}$, respectively. Let
$\gamma_{m_{1}}$, $\gamma_{m_{2}}$, and $\gamma_{m}$ be the primitive
$m_{1}$-th, $m_{2}$-th and $m$-th roots of unity used in the
encoding algorithms of
$\textbf{C}_{r}$, $\textbf{C}_{l}$, and $\textbf{C}_{r+l}$, respectively. It is
not hard to see that there are integers $\mu$ and $\nu$ such that
$\gamma_{m_{1}}=\gamma_{m}^{\mu m_{2}}$ and
$\gamma_{m_{2}}=\gamma_{m}^{\nu m_{1}}$. \citet{ItohSuzuki:2010}
proved that $P(X)=P_{1}(X^{\mu m_{2}})P_{2}(X^{\nu m_{1}})\in
\mathbb{F}_{2^{t}}[X]$ is an $S_{m}$-decoding polynomial for
$\textbf{C}_{r+l}$. Obviously, $P(X)$ contains at most
$k_{r}k_{l}$ monomials. Hence, the composition theorem below
follows.

\begin{theorem}[\citep{ItohSuzuki:2010}] \label{Thm:Itoh1}
With notations as above, there is a $k$-query linear LDC
$\emph{\textbf{C}}:\mathbb{F}_{2^{t}}^{n}\rightarrow
\mathbb{F}_{2^{t}}^{N_{r+l}}$ for which  $k\leq k_{r}k_{l}$.
\end{theorem}

\ref{Thm:Itoh1} shows that
Efremenko's LDC $\textbf{C}_{r+l}$ essentially has a local decoding algorithm which makes at most
$k_{r}k_{l}$ queries.
The key idea of the composition method is as follows: if we choose
the building blocks $\textbf{C}_{r}$ and $\textbf{C}_{l}$ in such
a way that either $k_{r}<2^{r}$ or $k_{l}<2^{l}$, then a local decoding algorithm for
$\textbf{C}_{r+l}$ which makes less than $2^{r+l}$ queries follows.
%
%
%
 For every integer $r\geq 4$,  applying
\ref{Thm:Itoh1} to Efremenko's 3-query  LDC $\textbf{C}_{2}$ (based
on $m_{1}=511$) of length $N_{2}$ and $k_{r-2}$-query LDC
$\textbf{C}_{r-2}$ (based on $m_{2}=q_{1}\ldots q_{r-2}$ such that
$\gcd(m_{1},m_{2})=1$) of length $N_{r-2}$ gives:

\begin{corollary}[\citep{ItohSuzuki:2010}] \label{Thm:Itoh2}
For every integer $r\geq 4$, there is a $k$-query linear LDC
$\emph{\textbf{C}}$ of length $N_{r}$ in which $k\leq 3\cdot
2^{r-2}$.
\end{corollary}

We note that  Efremenko's 3-query linear LDC
is crucial to the improvement provided by
\ref{Thm:Itoh2}. The existence of this code depends on a
carefully chosen good composite number $m_{1}=511$. It is natural to ask
whether there are good composite numbers other than 511 based on which a
3-query linear LDC of length $N_{2}$ can be obtained from Efremenko's
construction.

For every positive integer $r \geq 2$, we denote
by $\mathbb{M}_{r}$  the set of  integers, each of which is a
product of $r$ distinct odd primes and can yield a $k$-query
linear LDC of length $N_{r}$ for which $k<2^{r}$ in Efremenko's construction.
\citet{Efremenko:2009}  showed that $511\in
\mathbb{M}_{2}$ and built their 3-query LDC on this number.
\citet{ItohSuzuki:2010} proved that $15\not\in\mathbb{M}_{2}$ by exhaustive
search. Both \citet{Efremenko:2009} and \citet{ItohSuzuki:2010} left as an open problem to find
elements of $\mathbb{M}_{2}$ other than 511. We provide an answer to this problem in
the next section.

We end this section with some algebra required to establish our results.

\subsection{Group Rings, Characters and Cyclotomic Cosets}

Let $G$ be a finite multiplicative abelian group. The \emph{group ring}
\begin{equation*}
\mathbb{Z}[G]=\left\{\sum_{g\in G}a_{g}g: a_{g}\in \mathbb{Z}\right\}
\end{equation*}
is a ring of formal sums, in
which  addition
 and multiplication are defined as follows:
\begin{align*}
A+B &=\sum_{g\in G}(a_{g}+b_{g})g, \\
A\cdot B &=\sum_{g\in G}\sum_{h\in G}a_{g}b_{h}gh,
\end{align*}
where $A=\sum_{g\in G}a_{g}g,B=\sum_{g\in G}b_{g}g\in \mathbb{Z}[G]$.
The following are standard notations:
\begin{align*}
A^{(j)} &=\sum_{g\in G}a_{g}g^{j},~~~~~\forall j\in \mathbb{Z}, \\
D &=\sum_{g\in D}g,~~~~~
\forall D\subseteq G.
\end{align*}

Let $\mathbb{C}$ be the field of complex numbers and
$\mathbb{C}^{*}$ its multiplicative group. Any group homomorphism
$\chi:G\rightarrow \mathbb{C}^{*}$ is called a \emph{character} of
$G$. If $|G|=n$, then it has exactly $n$ distinct characters. Let
$\widehat{G}$ be the set of all characters of $G$. Then
$\widehat{G}$ is a multiplicative group in which
$\chi_{1}\chi_{2}(g)=\chi_{1}(g)\chi_{2}(g)$ for all
$\chi_{1},\chi_{2}\in \widehat{G}, g\in G$. The identity $\chi_{0}$ of $\widehat{G}$,
called the \emph{principal character}, maps every $g\in G$ to $1\in
\mathbb{C}^{*}$. For every $\chi\in \widehat{G}$, the \emph{order} of
$\chi$ is defined to be the least positive integer $l$ such that
$\chi^{l}=\chi_{0}$. Every $\chi\in \widehat{G}$ can be easily
extended to $\mathbb{Z}[G]$ linearly: $\chi(A)=\sum_{g\in
G}a_{g}\chi(g)$. The following properties are well-known:
\begin{enumerate}
\item If $|G|=n<\infty$, then for any  $\chi\in \widehat{G}$  and
 $g\in G$, $\chi(g)^{n}=1$.
  \item If $\chi\in \widehat{G}\setminus\{\chi_{0}\}$, then $\sum_{g\in
  G}\chi(g)=0$.
  \item $\chi(A^{(-1)})=\overline{\chi(A)}$, for every $\chi \in \widehat{G}, A\in
  \mathbb{Z}[G]$.
\end{enumerate}

Let $p$ be a prime or prime power and $m\in \mathbb{Z}^{+}$
such that $\gcd(p,m)=1$. For every $s\in \mathbb{Z}_{m}$, the
\emph{cyclotomic coset of $p$ modulo $m$ containing $s$} is defined
to be the following set
$$E_{s}=\{(sp^{l} \bmod m)\in \mathbb{Z}_{m}: l=0,1,\ldots\},$$
where $s$ is called  \emph{coset representative} of
$E_{s}$. We always suppose that  $s$ is
 smallest in $E_{s}$.
 It is well-known that all distinct cyclotomic cosets
of $p$ modulo $m$ form a partition of $\mathbb{Z}_{m}$.

The interested reader is referred to \citet{CurtisReiner:2006,Washington:1997,MacWilliamsSloane:1977,McDonald:1974}
for more information.

\section{\boldmath Mersenne Numbers which are Products of Two Primes Belong to $\mathbb{M}_{2}$}
\label{Sec:characterization}

In this section, we answer the open problem raised by \citet{Efremenko:2009} and
\citet{ItohSuzuki:2010} by proving that any Mersenne number which is
the product of two primes belongs to $\mathbb{M}_{2}$. This result allows us
to obtain a family of numbers in $\mathbb{M}_{2}$. Furthermore, we also give  characterizations
of numbers in $\mathbb{M}_{2}$, which turn out to be helpful for deciding whether
a given number is in $\mathbb{M}_{2}$.

Let $m=pq$ be the product of two distinct odd primes $p$ and $q$.
 Let $t$ be the multiplicative order of
2 in $\mathbb{Z}_{m}^{*}$, and let $\gamma_{m}\in
\mathbb{F}_{2^{t}}^{*}$ be a primitive $m$-th root of unity. Let
$S_{m}=\{s_{11}=1,s_{01},s_{10}\}$ be the canonical set of $m$.
 Then the set
of $S_{m}$-decoding polynomials is
\begin{equation*}
\mathcal{F}=\left\{f(X)\in
\mathbb{F}_{2^{t}}[X] : \text{$
f(\gamma_{m})=f(\gamma_{m}^{s_{01}})=f(\gamma_{m}^{s_{10}})=0$ and $f(1)=1$}\right\}.
\end{equation*}
By Lagrange interpolation, there exists $f\in \mathcal{F}$ that
contains at most four monomials.  On the other hand, we have the following proposition.

\begin{proposition}
\label{Lem:New100}
Let $m=pq$ be the product of two distinct odd primes. Then any
$S_{m}$-decoding polynomial contains at least three monomials.
\end{proposition}

\begin{proof}
Suppose $f(X)=ax^{u}+bx^{v}\in \mathcal{F}$ is an $S_{m}$-decoding
polynomial with less than three monomials. Then
$a\gamma_{m}^{u}+b\gamma_{m}^{v}=a\gamma_{m}^{us_{01}}
+b\gamma_{m}^{vs_{01}}=a\gamma_{m}^{us_{10}}+b\gamma_{m}^{vs_{10}}=0$
and $a+b=1$. It follows that
$a\gamma_{m}^{u-v}=a\gamma_{m}^{(u-v)s_{01}}=a\gamma_{m}^{(u-v)s_{10}}=1+a$.
Obviously, $a\neq 0$ and therefore
$\gamma_{m}^{u-v}=\gamma_{m}^{(u-v)s_{01}}=\gamma_{m}^{(u-v)s_{10}}$.
This implies that $m|\gcd((u-v)(s_{01}-1),
(u-v)(s_{10}-1),(u-v)(s_{10}-s_{01}))$. Since $\gcd(m,
s_{10}-s_{01})=1$, we have $m|(u-v)$. Hence,
$a=a\gamma_{m}^{u-v}=a\gamma_{m}^{(u-v)s_{01}}=a\gamma_{m}^{(u-v)s_{10}}=1+a$,
which is a contradiction.
\end{proof}

\ref{Lem:New1} shows that for $m=pq$, the best we can expect
is to have an $S_{m}$-decoding polynomial with exactly three
monomials. Let
\begin{equation*}
\mathcal{G}=\left\{g(X)\in \mathbb{F}_{2^{t}}[X]: \text{
$g(\gamma_{m})=g(\gamma_{m}^{s_{01}})=g(\gamma_{m}^{s_{10}})=0$ and
$g(1)\neq 0$}\right\}.
\end{equation*}
Then we have the following result.

\begin{proposition}
\label{Lem:New2}
There is an $S_{m}$-decoding polynomial $f\in \mathcal{F}$ with
three monomials if and only if there is a polynomial $g\in\mathcal{G}$ with three monomials.
\end{proposition}

\begin{proof}
The forward implication is trivial, since $\mathcal{F}\subseteq
\mathcal{G}$. Let $g\in \mathcal{G}$ have exactly three monomials.
Then $f(X)=g(X)/g(1)\in \mathcal{F}$ contains the same number of
monomials as $g(X)$, namely three.
\end{proof}

By \ref{Lem:New2}, finding an $S_{m}$-decoding polynomial with
exactly three monomials is equivalent to finding a polynomial
$g(X)\in \mathcal{G}$ with exactly three monomials. Let $g(X)\in
\mathcal{G}$ be such a polynomial. Since $\mathcal{G}$ is closed
under  multiplication by elements of $\mathbb{F}_{2^{t}}\setminus\{0\}$, we
may suppose, without loss of generality,
that $g(X)=X^{u}+aX^{v}+b\in \mathbb{F}_{2^{t}}[X]$ for some
distinct  $u,v\in \mathbb{Z}_{m}\setminus\{0\}$
(only $g(1)$, $g(\gamma_{m})$, $g(\gamma_{m}^{s_{01}})$ and
$g(\gamma_{m}^{s_{10}})$ are concerned)
 and $a,b\in
\mathbb{F}_{2^{t}}\setminus\{0\}$. By the definition of
$\mathcal{G}$, the following conditions hold simultaneously:
\begin{align}
\label{Equ:New1}
\begin{pmatrix}
           \gamma_{m}^{us_{01}}   & \gamma_{m}^{vs_{01}}  & 1 \\
             \gamma_{m}^{us_{10}} & \gamma_{m}^{vs_{10}}  & 1 \\
             \gamma_{m}^{u} & \gamma_{m}^{v} & 1 \\
\end{pmatrix}
\begin{pmatrix}
1 \\
a \\
b \\
\end{pmatrix}
&=
\begin{pmatrix}
0 \\
0 \\
0 \\
\end{pmatrix}, \\
\label{Equ:New2}
1+a+b &\not= 0.
\end{align}
Conditions \ref{Equ:New1} and \ref{Equ:New2} shed much light on how to determine
elements of $\mathbb{M}_{2}$. A
computer search based on these conditions
shows  that the Mersenne numbers $M_{11}=2^{11}-1=2047$
and $M_{23}=2^{23}-1=8388607$  both belong to $\mathbb{M}_{2}$
(see Table 3.1
for the corresponding $S_{m}$-decoding polynomials).

\begin{table}
\label{Tab:NewMembers1}
\small
\begin{tabular}{|c|c|c|}
\hline
$m$ & $M_{11}=2^{11}-1=2047$ & $M_{23} = 2^{23}-1 = 8388607$  \\
$\mathbb{F}_{2^t}$ & $\mathbb{F}_{2^{11}}=\mathbb{F}_2[\gamma]/(\gamma^{11}+\gamma^2+1)$ & $\mathbb{F}_{2^{23}}=\mathbb{F}_2[\gamma]/(\gamma^{23}+\gamma^5+1)$ \\
$S_m$ & $\{s_{11}=1,s_{01}=713,s_{10}=1335\}$ & $\{s_{11}=1, s_{01} = 5711393, s_{10}= 2677215\}$ \\
$f(X)$ & $\gamma^{1485}X^{29} + \gamma^{694}X^{27} + \gamma^{118}$ & $\gamma^{6526329}X^{3526}+\gamma^{7574532}X^{3363}+\gamma^{2861754}$  \\
\hline
\end{tabular}
\caption{New elements $m$ determined to be in $\mathbb{M}_{2}$}
\end{table}

\ref{Thm:Itoh1} shows that the more numbers in $\mathbb{M}_{2}$ we find, the more
improvements we get on the query complexity within Efremenko's framework. This motivates
the consideration of numbers taking the form of $M_{11}$ and $M_{23}$, and to
understand why they
yield better local decoding algorithms within Efremenko's framework. We note that
$M_{11}$ and $M_{23}$ are both Mersenne numbers and each a product of two primes.
This begs the question: do all numbers of this form belong to $\mathbb{M}_{2}$, and do they
intrinsically yield better local decoding algorithms in Efremenko's framework?
For the remaining of this section, we provide an affirmative answer to this question.
More precisely, we prove the following theorem.

\begin{theorem}
\label{Thm:New2}
Let $m=2^{t}-1=pq$ be a Mersenne number, where $t$, $p$ and $q$ are primes.
Then $m\in \mathbb{M}_{2}$.
\end{theorem}
The proof of \ref{Thm:New2} is  based on  analysis
of conditions \ref{Equ:New1} and \ref{Equ:New2}, and is an easy consequence of
Propositions \bare\ref{Lem:New1} and \bare\ref{Lem:New3} below.

\begin{proposition}
\label{Lem:New1}
Let $m=pq$ be the product of two distinct odd primes $p$ and $q$. Let
$t$ be the multiplicative order of $2\in\mathbb{Z}_{m}^{*}$, and let
$\gamma_{m}\in \mathbb{F}_{2^{t}}^{*}$ be a primitive $m$-th root of
unity. Define
\begin{equation}
\label{Equ:New4}
\mathcal{Z}=\left\{\frac{z_{1}+z_{2}}{z_{1}z_{2}+z_{2}}: \text{
$z_{1},z_{2}\in \mathbb{F}_{2^{t}}^{*}$,
${\rm ord}(z_{1})=p$, and ${\rm ord}(z_{2})=q$}\right\}.
\end{equation}
If $\mathcal{Z}$ is a multiset containing an element of multiplicity greater than
one, then $m\in\mathbb{M}_{2}$.
\end{proposition}

\begin{proof}
Suppose $\mathcal{Z}$ contains an element
of multiplicity greater than one. Then there exist  $z_{1},z_{2},z'_{1}, z'_{2}\in
\mathbb{F}_{2^{t}}^{*}$ such that the following hold:
\begin{enumerate}[(i)]
\label{Equ:New7}
\item ${\rm ord}(z_{1})={\rm ord}(z'_{1})=p$,
\item ${\rm ord}(z_{2})={\rm ord}(z'_{2})=q$,
\item $(z_{1},z_{2}) \neq (z'_{1},z'_{2})$,
\item $\dfrac{z_{1}+z_{2}}{z_{1}z_{2}+z_{2}} = \dfrac{z'_{1}+z'_{2}}{z'_{1}z'_{2}+z'_{2}}$.
\end{enumerate}
Obviously, we have
${\rm ord}(\gamma_{m}^{s_{10}})=p$ and ${\rm ord}(\gamma_{m}^{s_{01}})=q$. It follows that there are
integers $u_{1},v_{1}\in \mathbb{Z}_{p}\setminus\{0\}$ and $u_{2},v_{2}\in
\mathbb{Z}_{q}\setminus\{0\}$ such that the following hold:
\begin{enumerate}[(i)]
\setcounter{enumi}{4}
\label{Equ:New8}
\item $z_{1}=(\gamma_{m}^{s_{10}})^{u_{1}}=\gamma_{m}^{u_{1}s_{10}}$,
\item $z_{2}=\gamma_{m}^{u_{2}s_{01}}$,
\item $z'_{1}=\gamma_{m}^{v_{1}s_{10}}$,
\item $z'_{2}=\gamma_{m}^{v_{2}s_{01}}$.
\end{enumerate}
Since $p$ and $q$ are distinct primes, the Chinese Remainder Theorem implies that there are
unique numbers $u,v\in
\mathbb{Z}_{m}\setminus\{0\}$ such that
\begin{enumerate}[(i)]
\setcounter{enumi}{8}
\label{Equ:New9}
\item $u\equiv u_{1}\bmod{p}$ and $u\equiv u_{2}\bmod{q}$,
\item $v\equiv v_{1}\bmod{p}$ and $v\equiv v_{2}\bmod{q}$.
\end{enumerate}
Combing the set of conditions (i)--(x),
it is easy to verify that the numbers $u,v\in\mathbb{Z}_{m}\setminus\{0\}$ satisfy the
following conditions
\begin{enumerate}[(i)]
\setcounter{enumi}{10}
\label{Equ:New10}
\item $z_{1}=\gamma_{m}^{us_{10}}$, $ z_{2}=\gamma_{m}^{us_{01}}$,
$z'_{1}=\gamma_{m}^{vs_{10}}$, and $z'_{2}=\gamma_{m}^{vs_{01}}$,
\item $u\neq v$,
\item $\dfrac{\gamma_{m}^{u}+\gamma_{m}^{us_{01}}}{\gamma_{m}^{u}+
\gamma_{m}^{us_{10}}}=\dfrac{\gamma_{m}^{v}+\gamma_{m}^{vs_{01}}}{\gamma_{m}^{v}+\gamma_{m}^{vs_{10}}}$.
\end{enumerate}
The last condition (xiii) implies that
the  matrix
\begin{equation}
\label{Equ:New3}
\Gamma_{u,v}=
\begin{pmatrix}
           \gamma_{m}^{us_{01}}   & \gamma_{m}^{vs_{01}}  & 1 \\
             \gamma_{m}^{us_{10}} & \gamma_{m}^{vs_{10}}  & 1 \\
             \gamma_{m}^{u} & \gamma_{m}^{v} & 1 \\
\end{pmatrix}
\end{equation}
has determinant zero.
It follows that
${\rm rank}(\Gamma_{u,v})=1$ or 2. If ${\rm rank}(\Gamma_{u,v})=1$, then the rank of
\begin{equation*}
\begin{pmatrix}
           \gamma_{m}^{us_{01}}+\gamma_{m}^{u}   & \gamma_{m}^{vs_{01}}+\gamma_{m}^{v}  & 0 \\
             \gamma_{m}^{us_{10}}+ \gamma_{m}^{u} &\gamma_{m}^{vs_{10}}+\gamma_{m}^{v}  & 0 \\
             \gamma_{m}^{u}& \gamma_{m}^{v} & 1 \\
\end{pmatrix}
\end{equation*}
is also 1. Hence,
$\gamma_{m}^{us_{01}}+\gamma_{m}^{u}=\gamma_{m}^{vs_{01}}+\gamma_{m}^{v}=\gamma_{m}^{us_{10}}+
\gamma_{m}^{u}=\gamma_{m}^{vs_{10}}+\gamma_{m}^{v}=0$, which in turn implies
 $\gamma_{m}^{us_{01}}=\gamma_{m}^{us_{10}}$
 and
$\gamma_{m}^{vs_{01}}=\gamma_{m}^{vs_{10}}$. Since $\gamma_{m}$ is of order $m$
and $\gcd(m,s_{01}-s_{10})=1$, we have
$m|\gcd(u(s_{01}-s_{10}), v(s_{01}-s_{10}))$ and
therefore
$m|\gcd(u,v)$, which contradicts the fact that $u,v\in
\mathbb{Z}_{m}\setminus \{0\}$.
Consequently, ${\rm rank}(\Gamma_{u,v})=2$ and
the equation \ref{Equ:New1} has a unique solution $(a,b)\in \mathbb{F}_{2^{t}}^{2}$.

Next we show that both $a$ and $b$ are nonzero.  If $a=0$, then
$b=\gamma_{m}^{us_{01}}=\gamma_{m}^{us_{10}}=\gamma_{m}^{u}$, which
implies that $u\equiv 0\bmod{m}$. If $b=0$, then
$a=\gamma_{m}^{(u-v)s_{01}}=\gamma_{m}^{(u-v)s_{10}}=\gamma_{m}^{u-v}$,
which implies that $u\equiv v\bmod{m}$. Both cases yield
contradictions, since $u,v\in \mathbb{Z}_{m}\setminus \{0\}$ are distinct.

Let $g(X)=X^{u}+aX^{v}+b\in \mathbb{F}_{2^{t}}[X]$. Then $g(X)$
contains three monomials since $u,v\in \mathbb{Z}_{m}\setminus\{0\}$ are distinct
and $a,b\in \mathbb{F}_{2^{t}}\setminus\{0\}$.
Furthermore, we have
$g(\gamma_{m})=g(\gamma_{m}^{s_{01}})=g(\gamma_{m}^{s_{10}})=0$ since $(a,b)$
satisfies \ref{Equ:New1}.

As the last step, we claim that $g(1)\neq 0$, for otherwise
the vector $(1,1,1)$ is necessarily
a linear combination of the rows of $\Gamma_{u,v}$, since
$(1,a,b)\neq (0,0,0)$, and thereby
\begin{equation*}
\begin{pmatrix}
           \gamma_{m}^{us_{01}}   & \gamma_{m}^{vs_{01}}  & 1 \\
             \gamma_{m}^{us_{10}} & \gamma_{m}^{vs_{10}}  & 1 \\
             \gamma_{m}^{u} & \gamma_{m}^{v} & 1 \\
                                 1 & 1 & 1 \\
\end{pmatrix}
\end{equation*}
has rank two.
Applying elementary row operations (adding the third row to each of
the first three rows) to the above matrix gives
 \begin{equation}
 \label{Equ:New13}
\frac{1+\gamma_{m}^{u}}{1+\gamma_{m}^{v}}
=\frac{1+\gamma_{m}^{us_{10}}}{1+\gamma_{m}^{vs_{10}}}=
\frac{1+\gamma_{m}^{us_{01}}}{1+\gamma_{m}^{vs_{01}}}.
\end{equation}
Condition (xiii) and \ref{Equ:New13} now jointly yield
$\gamma_{m}^{(u-v)s_{01}}=\gamma_{m}^{(u-v)s_{10}}$, which in turn implies that
 $u=v$. This is a contradiction.

We have actually shown that
$g(X)\in \mathcal{G}$ and contains exactly three monomials.
By \ref{Lem:New2}, there is an
$S_{m}$-decoding polynomial $f(X)\in \mathcal{F}$ which also contains exactly three
monomials. Hence, $m\in \mathbb{M}_{2}$.
\end{proof}

\begin{proposition}
 \label{Lem:New3}
Let $m=2^{t}-1=pq$ be a Mersenne number, where $t$, $p$ and $q$ are all
primes, $p\not=q$. Then $\mathcal{Z}$ (as defined in \ref{Lem:New1}) is a multiset
containing an element of multiplicity greater than one.
\end{proposition}

\begin{proof}
Obviously, $\mathcal{Z}$ has at most $(p-1)(q-1)$ distinct elements.
Suppose $\mathcal{Z}$ is a set of cardinality $(p-1)(q-1)$. For
every $z_{1},z_{2}\in \mathbb{F}_{2^{t}}^{*}$ such that
${\rm ord}(z_{1})=p$ and ${\rm ord}(z_{2})=q$, we have
$(z_{1}+z_{2})/(z_{1}z_{2}+z_{2})=1+(1+z_{2}^{-1})/(1+z_{1}^{-1})$.
Hence,
\begin{equation}
\label{Equ:New14}
S=\left\{(1+z_{2})/(1+z_{1}): \text{$z_{1},z_{2}\in \mathbb{F}_{2^{t}}^{*}$,
${\rm ord}(z_{1})=p$, and ${\rm ord}(z_{2})=q$}\right\}
\end{equation}
is also a set of
cardinality $(p-1)(q-1)$. Let $G=\mathbb{F}_{2^{t}}^{*}$  and
$1_{G}$  its identity. Consider the group ring $\mathbb{Z}[G]$.  We
identify the two subsets of $G$,
\begin{align}
\label{Equ:New15}
A &=\{1+z_{1}: \text{$z_{1}\in\mathbb{F}_{2^{t}}^{*}$ and ${\rm ord}(z_{1})=p$}\}, \\
B &=\{1+z_{2}: \text{$z_{2}\in\mathbb{F}_{2^{t}}^{*}$ and ${\rm ord}(z_{2})=q$}\},
\end{align}
with two elements of $\mathbb{Z}[G]$.

We claim that
\begin{equation}\label{Equ:New16}
S\cup A^{(-1)}\cup B\cup\{1_{G}\}=G.
\end{equation}
Indeed,
since $S\cup A^{(-1)}\cup B\cup\{1_{G}\}\subseteq G$ and
$|S|+|A^{-1}|+|B|+|\{1_{G}\}|=|G|$, it suffices to show that
$S$, $A^{(-1)}$, $B$, and $\{1_{G}\}$ are pairwise disjoint. It is obvious that
$1_{G}\notin S\cup A^{(-1)}\cup B$.
If $S\cap A^{(-1)}\neq \varnothing$, then there exist
$z_{1},z'_{1},z_{2}\in \mathbb{F}_{2^{t}}^{*}$ such that
$(1+z_{2})/(1+z_{1})=1/(1+z'_{1})$, where
${\rm ord}(z_{1})={\rm ord}(z'_{1})=p$ and ${\rm ord}(z_{2})=q$. It follows that
$(1+z_{2}^{2})/(1+z_{1})=(1+z_{2})/(1+z'_{1})$, which contradicts
our assumption that $S$ is a set of cardinality $(p-1)(q-1)$.
Similarly, we have $S\cap B=A^{(-1)}\cap B=\varnothing$.

From \ref{Equ:New16} we derive
\begin{equation}\label{Equ:New17}
(A+1_{G})^{(-1)}(B+1_{G})=G.
\end{equation}
Let $\gamma_{p}, \gamma_{q}\in G$ be some primitive
$p$-th and $q$-th roots of unity, respectively. We claim that
there exist a permutation $a:\mathbb{Z}_{p}^{*}\rightarrow
\mathbb{Z}_{p}^{*}$ and a mapping $b:\mathbb{Z}_{p}^{*}\rightarrow
\mathbb{Z}_{q}$ such that for every $i\in \mathbb{Z}_{p}^{*}$,
\begin{equation}\label{Equ:New18}
1+\gamma_{p}^{i}=\gamma_{p}^{a(i)}\gamma_{q}^{b(i)}.
\end{equation}
Let $\theta_{p},\theta_{q}\in \mathbb{C}$ be some complex primitive
$p$-th and $q$-th roots of unity respectively, where $\mathbb{C}$ is
the field of complex numbers. Let $\chi_{p}$ be a multiplicative
character of order $p$ of the group $G$, such that
$\chi_{p}(\gamma_{p})=\theta_{p}$. The identity
$\chi_{p}((A+1_{G})^{(-1)})\chi_{p}(B+1_{G})=\chi_{p}(G)=0$ implies that
either $\chi_{p}((A+1_{G})^{(-1)})=0$ or
$\chi_{p}(B+1_{G})=0$.
If  $\chi_{p}(B+1_{G})=0$, then $q\equiv \chi_{p}(B+1_{G})\equiv
0\bmod{(1-\theta_{p})}$ and therefore $q\in
(1-\theta_{p})\mathbb{Z}[\theta_{p}]$. On the other hand,
$p=\Pi_{i=1}^{p-1}(1-\theta_{p}^{i})\in
(1-\theta_{p})\mathbb{Z}[\theta_{p}]$. Since $\gcd(p,q)=1$, there
are rational integers $\alpha,\beta$ such that $\alpha p+\beta
q=1$. It follows that $1\in (1-\theta_{p})\mathbb{Z}[\theta_{p}]$,
which contradicts the well-known fact that
$(1-\theta_{p})\mathbb{Z}[\theta_{p}]$ is a prime ideal in
$\mathbb{Z}[\theta_{p}]$ (cf. \citet[Lemma 1.4]{Washington:1997}).
Hence,
we have $\chi_{p}((A+1_{G})^{(-1)})=0$ and
$\chi_{p}(A+1_{G})=\overline{\chi_{p}((A+1_{G})^{(-1)})}=0$, giving
$\sum_{i=1}^{p-1}\chi_{p}(1+\gamma_{p}^{i})+1=0$. Clearly, there
is a mapping $a:\mathbb{Z}_{p}^{*}\rightarrow \mathbb{Z}_{p}$ such
that $\chi_{p}(1+\gamma_{p}^{i})=\theta_{p}^{a(i)}$ for all $i\in
\mathbb{Z}_{p}^{*}$. Hence, $\sum_{i=1}^{p-1}\theta_{p}^{a(i)}+1=0$.
Since any $p-1$ elements of $\{1,\theta_{p},\ldots,
\theta_{p}^{p-1}\}$ form an integral basis of $\mathbb{Z}[\theta_{p}]$
over $\mathbb{Z}$, $a$ must be a permutation of
$\mathbb{Z}_{p}^{*}$. Since
$G=\{\gamma_{p}^{\alpha}\gamma_{q}^{\beta}: \alpha \in
\mathbb{Z}_{p},\beta\in \mathbb{Z}_{q}\}$, there are two mappings
$\alpha: \mathbb{Z}_{p}^{*}\rightarrow \mathbb{Z}_{p}$ and
$\beta:\mathbb{Z}_{p}^{*}\rightarrow \mathbb{Z}_{q}$ such that
$1+\gamma_{p}^{i}=\gamma_{p}^{\alpha(i)}\gamma_{q}^{\beta(i)}$ for
all $i\in \mathbb{Z}_{p}^{*}$ . It follows that
$\theta_{p}^{a(i)}=\chi_{p}(1+\gamma_{p}^{i})=\chi_{p}(\gamma_{p}^{\alpha(i)})
\chi_{p}(\gamma_{q}^{\beta(i)})=\theta_{p}^{\alpha(i)}\chi_{p}(\gamma_{q})^{\beta(i)}$.
Obviously, $\chi_{p}(\gamma_{q})^{p}=\chi_{p}(\gamma_{q})^{q}=1$ and so
 $\chi_{p}(\gamma_{q})=1$. Therefore,
$\theta_{p}^{a(i)}=\theta_{p}^{\alpha(i)}$, which implies
$\alpha=a$. We identify $\beta$ with $b$ and obtain
\ref{Equ:New18}.

Similarly, there exist a permutation $c:\mathbb{Z}_{q}^{*}\rightarrow
\mathbb{Z}_{q}^{*}$ and a mapping $d:\mathbb{Z}_{q}^{*}\rightarrow
\mathbb{Z}_{p}$ such that, for every $j\in \mathbb{Z}_{q}^{*}$,
\begin{equation}
\label{Equ:New19}
1+\gamma_{q}^{j}=\gamma_{q}^{c(j)}\gamma_{p}^{d(j)}.
\end{equation}

Let $\chi_{m}$ be a multiplicative character of order $m$ of $G$.
Without loss of generality, we suppose that
$\chi_{m}(\gamma_{p})=\theta_{p}$ and
$\chi_{m}(\gamma_{q})=\theta_{q}$. Applying $\chi_{m}$ to \ref{Equ:New17},
we have $\chi_{m}((A+1_{G})^{(-1)})\chi_{m}(B+1_{G})=\chi_{m}(G)=0$,
which implies either $\chi_{m}(A+1_{G})=0$ or $\chi_{m}(B+1_{G})=0$.
If $\chi_{m}(A+1_{G})=0$, then
$0=\sum_{i=1}^{p-1}\chi_{m}(1+\gamma_{p}^{i})+1=\sum_{i=1}^{p-1}\theta_{p}^{a(i)}\theta_{q}^{b(i)}+1
=\sum_{i=1}^{p-1}\theta_{p}^{a(i)}(\theta_{q}^{b(i)}-1)$. Since
$\{\theta_{p},\ldots,\theta_{p}^{p-1}\}$ is an integral basis of
$\mathbb{Z}[\theta_{p}, \theta_{q}]$ over $\mathbb{Z}[\theta_{q}]$,
we have $\theta_{q}^{b(i)}-1=0$ for every $i\in \mathbb{Z}_{p}^{*}$.
It follows that $1+\gamma_{p}^{i}=\gamma_{p}^{a(i)}$ for every $i\in
\mathbb{Z}_{p}^{*}$. Hence, $\{0,1,\gamma_{p},\ldots,
\gamma_{p}^{p-1}\}$ is a subfield of $\mathbb{F}_{2^{t}}$. However,
the only subfields of $\mathbb{F}_{2^{t}}$ are $\mathbb{F}_{2}$
and $\mathbb{F}_{2^{t}}$. Hence, either $p+1=2$ or
$p+1=2^{t}$, that is, either $p=1$ or $q=1$, which is a
contradiction.

Similarly, if $\chi_{m}(B+1_{G})=0$, then we conclude that
$\{0,1,\gamma_{q},\ldots,\gamma_{q}^{q-1}\}$ is a subfield of
$\mathbb{F}_{2^{t}}$, which yields the same contradiction.

Hence, our assumption that ${\mathcal Z}$ is a set of cardinality $(p-1)(q-1)$
is wrong and the proposition is established.
\end{proof}

We are now ready to proof \ref{Thm:New2}.

\begin{namedproof}{Proof of \ref{Thm:New2}}
To apply Propositions \bare\ref{Lem:New1} and \bare\ref{Lem:New3}, we need to
show that $p$ and $q$ are odd and distinct. Since $pq=m=2^{t}-1$ is odd, it suffices to show that
$p$ and $q$ are distinct. Suppose $p=q$, then $pq\equiv p^{2}\equiv 1\bmod{4}$
and $pq\equiv m\equiv 2^{t}
-1 \equiv -1\bmod{4}$, which is a contradiction.
\end{namedproof}

\ref{Thm:New2} provides a general method of obtaining new numbers in
$\mathbb{M}_{2}$ and  motivates the following definition of a subset of $\mathbb{M}_{2}$:
\begin{equation*}
\mathbb{M}_{2,{\rm Mersenne}} =
\left\{m : m=2^{t}-1=pq, \text{ where $t$, $p$ and $q$ are primes}\right\}.
\end{equation*}
It is an interesting open problem to determine the
cardinality of $\mathbb{M}_{2,{\rm Mersenne}}$. A similar but much more
well-known problem in number theory is determining
the number of Mersenne primes. Although
it is generally believed that there are infinitely many Mersenne primes,
no proof or disproof is known. It seems that our question on the cardinality of
$\mathbb{M}_{2,\rm Mersenne}$ is also difficult to answer.
We have, however, determined 50 elements of $\mathbb{M}_{2,\rm Mersenne}$
by computer search.
These fifty numbers $M_{t}=2^{t}-1=pq\in\mathbb{M}_{2,{\rm Mersenne}}$
with their smaller prime divisors $p$ are
listed in Table 3.2.
\begin{sidewaystable}
\label{Tab:NewMembers2}
\resizebox{\textheight}{!}{
\begin{tabular}{r|l|r|l}
\toprule
$m$ & $p$ & $m$ & $p$ \\
\toprule
$M_{11}$ & 23 & $M_{373}$ & 25569151 \\
$M_{23}$ & 47 & $M_{379}$ & 180818808679 \\
$M_{37}$ & 223 & $M_{421}$ & 614002928307599 \\
$M_{41}$ & 13367 & $M_{457}$ & 150327409 \\
$M_{59}$ & 179951 & $M_{487}$ & 4871 \\
$M_{67}$ & 193707721 & $M_{523}$ & 160188778313202118610543685368878688932828701136501444932217468039063 \\
$M_{83}$ & 167 & $M_{727}$ & 176062917118154340379348818723316116707774911664453004727494494365756 \\\
& & & 22328171096762265466521858927 \\
$M_{97}$ & 11447 & $M_{809}$ & 4148386731260605647525186547488842396461625774241327567978137 \\
$M_{101}$ & 7432339208719 & $M_{881}$ & 26431 \\
$M_{103}$ & 2550183799 & $M_{971}$ & 23917104973173909566916321016011885041962486321502513 \\
$M_{109}$ & 745988807 & $M_{983}$ & 1808226257914551209964473260866417929207023 \\
$M_{131}$ & 263 & $M_{997}$ & 167560816514084819488737767976263150405095191554732902607 \\
$M_{137}$ & 32032215596496435569 & $M_{1063}$ &1485761479 \\
$M_{139}$ & 5625767248687 & $M_{1427}$ & 19054580564725546974193126830978590503
\\
$M_{149}$ & 86656268566282183151 & $M_{1487}$ & 24464753918382797416777 \\
$M_{167}$ & 2349023 & $M_{1637}$ & 81679753 \\
$M_{197}$ & 7487 & $M_{2927}$ & 1217183584262023230020873 \\
$M_{199}$ & 164504919713 & $M_{3079}$ & 25324846649810648887383180721 \\
$M_{227}$ & 26986333437777017 & $M_{3259}$ & 21926805872270062496819221124452121 \\
$M_{241}$ & 22000409 & $M_{3359}$ & 6719 \\
$M_{269}$ & 13822297 & $M_{4243}$ & 101833 \\
$M_{271}$ & 15242475217 & $M_{4729}$ & 61944189981415866671112479477273 \\
$M_{281}$ & 80929 & $M_{5689}$ & 919724609777 \\
$M_{293}$ & 40122362455616221971122353 & $M_{6043}$ & 11155520642419038056369903183 \\
$M_{347}$ & 14143189112952632419639 & $M_{7331}$ & 458072843161 \\
\bottomrule
\end{tabular}
}
\caption{Fifty elements in $\mathbb{M}_{2,{\rm Mersenne}}$}
\end{sidewaystable}
The first 33
numbers in $\mathbb{M}_{2,{\rm Mersenne}}$ are
$M_{11},M_{23},\ldots,M_{809}$. However, we do not know whether
$M_{881}$ is the $34$th number in $\mathbb{M}_{2,{\rm Mersenne}}$ or
not.

We summarize our results below.

\begin{proposition}
\label{Cor:New3}
$|\mathbb{M}_{2,{\rm Mersenne}}|\geq 50$.
\end{proposition}

It seems  reasonable to conjecture that
$|\mathbb{M}_{2, \rm Mersenne}|=\infty$.

The set $\mathbb{M}_{2,\rm Mersenne}$ does enable us to improve
query complexity in Efremenko's framework through Itoh and Suzuki's
composition method (\ref{Thm:Itoh1}). However, to apply this method,
we have to make sure that the elements of $\mathbb{M}_{2,\rm Mersenne}$
are pairwise relatively prime.

\begin{proposition}
\label{Lem:New4}
 (a) Any two distinct elements in
$\mathbb{M}_{2,{\rm Mersenne}}$ are relatively prime. (b) Elements in $\mathbb{M}_{2,{\rm Mersenne}}$ are relatively prime to $511$.
\end{proposition}

\begin{proof}
(a) Let $M_{t}=2^{t}-1=pq\in \mathbb{M}_{2, \rm Mersenne}$ and
let $t_{1}$ and $t_{2}$ be the multiplicative orders of $2$ in
$\mathbb{Z}_{p}^{*}$ and $\mathbb{Z}_{q}^{*}$, respectively.
Then $t_{1}|t$ and $ t_{2}|t$, which in turn implies
$t_{1}=t_{2}=t$ since $t$ is prime and $t_{1},t_{2}>1$. Suppose  there are
two distinct numbers $M_{t},M_{t'}\in \mathbb{M}_{2, \rm Mersenne}$
such that $\gcd(M_{t}, M_{t'})>1$. Then  $M_{t}$ and $ M_{t'}$
have a common prime factor, say $p$. It follows that  $t=t'={\rm ord}_{p}(2)$, the
multiplicative order of $2\in \mathbb{Z}_{p}^{*}$.
Hence, we have $M_{t}=M_{t'}$, which is  a contradiction.

(b) Suppose that $M_{t}=2^{t}-1\in \mathbb{M}_{2,{\rm
Mersenne}}$ is such that $\gcd(M_{t},511)>1$.
Then either $7|M_{t}$ or $73|M_{t}$. The
multiplicative orders of 2 in $\mathbb{Z}_{7}^{*}$ and
$\mathbb{Z}_{73}^{*}$ are 3 and 9 respectively. Hence,
$3|t$ or $9|t$. However, $t$ is prime and greater than 9, which yields a
contradiction.
\end{proof}

The result below follows from Propositions \bare\ref{Cor:New3} and \bare\ref{Lem:New4}.

\begin{corollary}\label{Cor:New4}
There are at least $51$ elements in $\mathbb{M}_{2}$ which are
pairwise relatively prime.
\end{corollary}

Although \ref{Thm:New2} provides a rather general method of finding new elements in
$\mathbb{M}_{2}$ (since $\mathbb{M}_{2,{\rm Mersenne}}\subset \mathbb{M}_{2}$),
it does not provide a way for disproving membership
in $\mathbb{M}_{2}$ that is easier than exhaustive search.
\citet{ItohSuzuki:2010}
showed that $15\not\in\mathbb{M}_{2}$ by
exhaustive search. The next result shows that it is possible to avoid exhaustive
search in proving that $15\not\in\mathbb{M}_2$.

\begin{proposition}
\label{Lem:NegRes}
Let $p$, $q$, $m$, $t$, $\gamma_{m}$, and $\mathcal{Z}$ be as defined in
\ref{Lem:New1}. Then $m\in \mathbb{M}_{2}$ if
and only if there are cyclotomic
cosets  $E_{\alpha}$ and $E_{\beta}$ of $2$ modulo $m$ ($\alpha, \beta
\in \mathbb{Z}_{m}$) such that $ E_{\alpha}\cup E_{\beta}$ does not
contain any multiples of $p$ or $q$ and nonnegative integers $c,d<t$ such that
\begin{align}
\label{Equ:New5}
(\alpha,c) &\neq (\beta, d), \\
\label{Equ:New5b}
 \left(\frac{\gamma_{m}^{\alpha}+ \gamma_{m}^{\alpha
s_{01}}}{\gamma_{m}^{\alpha}+ \gamma_{m}^{\alpha
s_{10}}}\right)^{2^{c}} &=\left(\frac{\gamma_{m}^{\beta}+\gamma_{m}^{\beta
s_{01}}}{\gamma_{m}^{\beta}+\gamma_{m}^{\beta
s_{10}}}\right)^{2^{d}}.
\end{align}
\end{proposition}

\begin{proof}
Suppose $m\in \mathbb{M}_{2}$. By
\ref{Lem:New100}, there is an $S_{m}$-decoding polynomial
$f(X)\in \mathcal{F}$ with exactly three monomials.  By
\ref{Lem:New2}, there is a $g(X)\in \mathcal{G}$ with
exactly three monomials. Without loss of generality, let
$u,v\in \mathbb{Z}_{m}\setminus\{0\}$ be distinct and
$a,b\in \mathbb{F}_{2^{t}}\setminus\{0\}$ be such that
 $g(X)=X^{u}+aX^{v}+b\in \mathbb{F}_{2^{t}}[X]$. It follows that
\ref{Equ:New1} and \ref{Equ:New2} hold, and therefore
$\det(\Gamma_{u,v})=0$, which in turn implies the following identity
\begin{equation}
\label{Equ:New6}
(\gamma_{m}^{u}+\gamma_{m}^{us_{01}})(\gamma_{m}^{v}+\gamma_{m}^{vs_{10}})
=(\gamma_{m}^{u}+\gamma_{m}^{us_{10}})(\gamma_{m}^{v}+\gamma_{m}^{vs_{01}}).
\end{equation}
Since all cyclotomic cosets of 2 modulo $m$ form
a partition of $\mathbb{Z}_{m}$,  there exist $\alpha,\beta\in
\mathbb{Z}_{m}$ such that $u\in E_{\alpha}$ and $v\in E_{\beta}$,
where $E_{\alpha}$ and $E_{\beta}$ are cyclotomic cosets of 2 modulo
$m$ with representatives $\alpha$ and $\beta$, respectively.

Suppose that $hp\in E_{\alpha}$ for some integer $h$. Then
$q\nmid h$, for otherwise $\alpha=0$ and therefore $u=0$, which is a
contradiction. Since $u\in E_{\alpha}$, there is an integer $l$
such that $u\equiv 2^{l}hp\bmod{m}$.
It follows that
$\gamma_{m}^{u}+\gamma_{m}^{us_{01}}=(\gamma_{m}^{hp}+\gamma_{m}^{hps_{01}})^{2^{l}}=0$
since $hps_{01}\equiv hp\bmod{m}$.
By identity \ref{Equ:New6}, we have
$(\gamma_{m}^{u}+\gamma_{m}^{us_{10}})(\gamma_{m}^{v}+\gamma_{m}^{vs_{01}})=0$.
Since $hps_{10}\neq hp\bmod{m}$, we have
$\gamma_{m}^{u}+\gamma_{m}^{us_{10}}=(\gamma_{m}^{hp}+\gamma_{m}^{hps_{10}})^{2^{l}}\neq
0$, which in turn implies that
$\gamma_{m}^{v}+\gamma_{m}^{vs_{01}}=0$ and therefore $p|v$.
Thus,
$\gamma_{m}^{us_{10}}=\gamma_{m}^{2^{l}hps_{10}}=(\gamma_{m}^{hps_{10}})^{2^{l}}=1$
and $\gamma_{m}^{vs_{10}}=(\gamma_{m}^{ps_{10}})^{v/p}=1$. In other
words, the second row of $\Gamma_{u,v}$ is $(1,1,1)$, which implies
$1+a+b=0$ by \ref{Equ:New1}, contradicting \ref{Equ:New2}.
Hence, $E_{\alpha}$ does not contain any multiples
of $p$.
Similarly, $E_{\alpha}$ does not contain any multiples of $q$ and
$E_{\beta}$ does not contain any multiples of $p$ or $q$.

For $u\in E_{\alpha}$ and $ v\in E_{\beta}$, there exist nonnegative
integers $c,d<t$ such that $u\equiv 2^{c}\alpha\bmod{m}$ and $v\equiv
2^{d}\beta\bmod{m}$. The fact that  $u\neq v$ implies  $(\alpha,c)\neq (\beta, d)$.
Let $u=2^{c}\alpha$ and $v=2^{d}\beta$ in
\ref{Equ:New6}. Then \ref{Equ:New5b} follows.

It remains to show that the converse is also true.
Let $u\equiv 2^{c}\alpha \bmod{m}$
and $v\equiv 2^{d}\beta\bmod{m}$.
Then $u, v\in \mathbb{Z}_{m}$ are nonzero and distinct.
Let $z_{1}=\gamma_{m}^{us_{10}}$, $z_{2}=\gamma_{m}^{us_{01}}$,
$z'_{1}=\gamma_{m}^{vs_{10}}$, and $ z'_{2}=\gamma_{m}^{vs_{01}}$. Then
it is easy to verify that ${\rm ord}(z_{1})={\rm ord}(z'_{1})=p$, ${\rm ord}(z_{2})={\rm ord}(z'_{2})=q$  and
$(z_{1},z_{2})\neq (z'_{1},z'_{2})$. Then \ref{Equ:New5b} implies
\begin{equation}\label{Equ:Colli}
(z_{1}+z_{2})/(z_{1}z_{2}+z_{2})=(z'_{1}+z'_{2})/(z'_{1}z'_{2}+z'_{2}).
\end{equation}
Note that \ref{Equ:Colli} shows that $\mathcal{Z}$ is a multiset which contains an element of
multiplicity greater than one. By \ref{Lem:New1}, we have $m\in \mathbb{M}_{2}$,
which completes the proof.
\end{proof}

\ref{Lem:NegRes} provides a  rough characterization of elements in $\mathbb{M}_{2}$.
However, it turns out to be helpful for proving that some integers are not in $\mathbb{M}_{2}$. In
particular, we obtain a computer-free proof of the following result of \citet{ItohSuzuki:2010}.

\begin{corollary}
\label{Cor:New1}
$15\notin \mathbb{M}_{2}$.
\end{corollary}

\begin{proof}
The multiplicative order of $2\in\mathbb{Z}_{15}^{*}$ is $t=4$, and
$S_{15}=\{1,6,10\}$. Let
$\mathbb{F}_{2^{4}}=\mathbb{F}_{2}[\gamma]/(\gamma^{4}+\gamma+1)$
and let $\gamma$ be a primitive 15-th root of unity. The cyclotomic cosets
of 2 modulo 15 are $E_{0}=\{0\}$, $E_{1}=\{1,2,4,8\}$,
$E_{3}=\{3,6,9,12\}$, $E_{5}=\{5,10\}$, and $E_{7}=\{7,14,13,11\}$. If
$15\in \mathbb{M}_{2}$, then by \ref{Lem:NegRes}, there are
cyclotomic cosets $E_{\alpha}$ and $E_{\beta}$ such that $E_{\alpha}\cup
E_{\beta}$ does not contain any multiples of three or five and nonnegative
integers $c,d<4$ such that \ref{Equ:New5} and \ref{Equ:New5b} hold. It follows that
$\{\alpha,\beta\}\subseteq \{1,7\}$.

If $\alpha=\beta=1$, then
$((\gamma+\gamma^{6})/(\gamma+\gamma^{10}))^{2^{c}}
  =((\gamma+\gamma^{6})/(\gamma+\gamma^{10}))^{2^{d}}$ by \ref{Equ:New5b}, that is,
   $\gamma^{3\cdot 2^{c}}=\gamma^{3\cdot 2^{d}}$. It follows
  that $c=d$ and therefore $(\alpha,c)=(\beta,d)$, which is a
  contradiction.

If $\alpha=\beta=7$, then
$((\gamma^{7}+\gamma^{42})/(\gamma^{7}+\gamma^{70}))^{2^{c}}
  =((\gamma^{7}+\gamma^{42})/(\gamma^{7}+\gamma^{70}))^{2^{d}}$  by \ref{Equ:New5b},
that is,  $\gamma^{11\cdot 2^{c}}=\gamma^{11\cdot 2^{d}}$. It follows
  that
  $c=d$ and thereby  $(\alpha,c)=(\beta,d)$, which is a
  contradiction.

If $\{\alpha,\beta\}=\{1,7\}$, then
$((\gamma+\gamma^{6})/(\gamma+\gamma^{10}))^{2^{c}}
  =((\gamma^{7}+\gamma^{42})/(\gamma^{7}+\gamma^{70}))^{2^{d}}$  by \ref{Equ:New5b},
that is,
  $\gamma^{3\cdot 2^{c}}=\gamma^{11\cdot 2^{d}}$. Since
  $\gcd(2^{c},15)=\gcd(2^{d},15)=1$, we have that
  ${\rm ord}(\gamma^{3})={\rm ord}(\gamma^{11})$. However, ${\rm ord}(\gamma^{3})=5\neq
  15={\rm ord}(\gamma^{11})$, which is a contradiction.
\end{proof}

\section{Improved LDCs and PIR Schemes}
\label{Sec:LdcAndPir}

In this section, we apply the set $\mathbb{M}_{2,\rm
Mersenne}$ to the constructions of
 LDCs and information-theoretic PIR schemes. Consequently, we obtain
 a new family of query-efficient LDCs and a new family of PIR
 schemes with few servers. Compared with previous results of
 \citet{Efremenko:2009} and \citet{ItohSuzuki:2010},
the new LDCs and PIR schemes do achieve quantitative improvements
of efficiency which are considerable.

\subsection{Query-Efficient Locally Decodable Codes}
\label{subSec:newLDC}

By \ref{Cor:New4}, Theorem \ref{Thm:Efr2},
\ref{Thm:Itoh1} and Table 3.1, we have the following theorem:

\begin{theorem} \label{Thm:New3}
Let $N_{r}=\exp(\exp(O(\sqrt[r]{\log n(\log \log n)^{r-1}})))$.  Then the
following statements hold:
\begin{enumerate}[(a)]
\item For every positive integer $r\leq 103$, there is a
$k$-query linear LDC of length $N_{r}$ for which
\begin{equation*}
k \leq \begin{cases}
(\sqrt{3})^r,&\text{if $r$ is even} \\
8\cdot (\sqrt{3})^{r-3},&\text{if $r$ is odd.}
\end{cases}
\end{equation*}
\item
For every integer $r\geq 104$, there is a $k$-query linear LDC of length
$N_{r}$ for which $k\leq (3/4)^{51}\cdot 2^{r}$.
\item
If $|\mathbb{M}_{2,{\rm Mersenne}}|=\infty$, then for every integer
$r\geq 1$, there is a $k$-query linear LDC of length $N_{r}$ for which
$k$ is the same as in (a).
\end{enumerate}
\end{theorem}

\begin{proof}
\begin{enumerate}[(a)]
\item Let $r\in[103]$ be even.  By \ref{Cor:New4}, we
can take distinct $m_{1},\ldots, m_{r/2}\in\mathbb{M}_{2}$
which are
pairwise relatively prime. There is a
3-query linear LDC of length $N_{2}$ based on each of them by the
definition of $\mathbb{M}_{2}$ and \ref{Thm:Efr2}. Applying
\ref{Thm:Itoh1} $r/2-1$ times, we obtain a $k$-query
linear LDC of length $N_{r}$ for which $k\leq 3^{r/2}$, that is,
$k\leq (\sqrt{3})^{r}$.

Let $r\in[103]$ be odd. If $r=1$, then the Hadamard code is a 2-query
linear LDC of length $N_{1}=\exp(n)$ satisfying the required condition. If
$r\geq 3$, then $r=2\cdot\frac{r-3}{2}+3$ and we can take
distinct $m_{1}, \ldots, m_{\frac{r-3}{2}}\in\mathbb{M}_{2}$ which are pairwise
relatively prime. Since there are infinitely many primes, we can
always take another $m_{\frac{r-1}{2}}$ to be
a product of three distinct odd primes such that $m_{\frac{r-1}{2}}$ is
relatively prime to all of $m_{1}, \ldots, m_{\frac{r-3}{2}}$. By
\ref{Thm:Efr2}, there are a 3-query linear LDC of length
$N_{2}$ based on each of $m_{1}, \ldots, m_{\frac{r-3}{2}}$ and a
$k_{3}$-query linear LDC of length $N_{3}$ for which $k_{3}\leq
2^{3}$. Applying \ref{Thm:Itoh1} $(r-3)/2$ times gives
a $k$-query linear LDC of length $N_{r}$ for which $k\leq
3^{\frac{r-3}{2}}\cdot 8=8\cdot (\sqrt{3})^{r-3}$.

\item If $r\geq 104$, we take distinct $m_{1}, \ldots,
m_{51}\in\mathbb{M}_{2}$ and
$m_{52}$ a product of $r-102$
distinct odd primes such that $\gcd(m_{i},m_{j})=1$ for all
distinct $i,j\in[52]$.
By \ref{Thm:Efr2}, there is a 3-query
linear LDC of length $N_{2}$ based on each of $m_{1}, \ldots,
m_{51}$ and a $k_{r-102}$-query linear LDC of length $N_{r-102}$
based on $m_{52}$. Application of \ref{Thm:Itoh1} gives
a $k$-query linear LDC of length $N_{r}$ for which
$k\leq 3^{51}\cdot 2^{r-102}=(3/4)^{51}\cdot
2^{r}$.

\item It suffices to prove the statement for $r\geq 104$. If  $r$ is
even, we take $r/2$ distinct elements from
$\mathbb{M}_{2,{\rm Mersenne}}$ and if $r$ is odd, we take
$(r-3)/2$ distinct elements from $\mathbb{M}_{2,{\rm
Mersenne}}$ together with $m$, a
product of three distinct odd primes such that $\gcd(m,m_{i})=1$
 for all $i\in[(r-3)/2]$. In both cases, an application
 of \ref{Thm:Itoh1} yields the required conclusion. \qed
 \end{enumerate}
 \end{proof}

\subsection{Private Information Retrieval Schemes with Fewer Servers}
\label{Sec:pir}

An important application of LDCs is in the construction of
information-theoretic PIR schemes. A
PIR scheme allows a user $\mathcal{U}$ to retrieve a data item
$x_{i}$ from a database $x=(x_{1},\ldots,x_{n})\in\{0,1\}^{n}$ while
keeping the identity $i$ secret from the database operator. Since its introduction
by \citet{Choretal:1998}, many constructions have been
proposed
\citep{Choretal:1998,Ambainis:1997,Itoh:1999,Beimeletal:2005,Beimeletal:2002,WoodruffYekhanin:2007,Yekhanin:2008,Raghavendra:2007,Efremenko:2009,ItohSuzuki:2010}.
The efficiency of a PIR scheme is mainly measured by its communication complexity.
In this section, we turn our new query-efficient LDCs into PIR schemes that are more efficient
than those of \citet{Efremenko:2009} and \citet{ItohSuzuki:2010}.

\begin{definition}[PIR Scheme]
A one-round $k$-server PIR scheme  is a triplet of
algorithms $\mathcal{P}=(\mathcal{Q,A,C})$, where $\mathcal{Q}$ is a
probabilistic query algorithm, $\mathcal{A}$ is an answer algorithm,
and $\mathcal{C}$ is a reconstruction algorithm. At the beginning of
the scheme, $\mathcal{U}$  picks a random string {\tt aux}, computes a $k$-tuple
of queries ${\tt que}=({\tt que}_{1},\ldots, {\tt que}_{k})=\mathcal{Q}(k,n,i,{\tt aux})$ and
sends each query ${\tt que}_{j}$ to server
$S_{j}$. After receiving  ${\tt que}_{j}$, the server $S_{j}$ replies to
$\mathcal{U}$ with
${\tt ans}_{j}=\mathcal{A}(k,n,j,x,{\tt que}_{j})$. At last,  $\mathcal{U}$  outputs
$\mathcal{C}(k,n,i,{\tt aux},{\tt ans}_{1},\ldots,{\tt ans}_{k})$ such that:
\begin{description}
\item[\sf Correctness:] For every  integer $n$,
 $ x\in \{0,1\}^{n}$, $i\in[n]$, and
 {\tt aux},
\begin{equation*}
\mathcal{C}(k,n,i,{\tt aux},{\tt ans}_{1},\ldots,{\tt ans}_{k})=x_{i}.
\end{equation*}
\item[\sf Privacy:]
For every $i_{1},i_{2}\in[n]$, $j\in [k]$, and query ${\tt que}$,
\begin{equation*}
\Pr[\mathcal{Q}_{j}(k,n,i_{1},{\tt aux})={\tt que}]
=\Pr[\mathcal{Q}_{j}(k,n,i_{2},{\tt aux})={\tt que}].
\end{equation*}
\end{description}
\end{definition}

The \textit{communication complexity} of $\mathcal{P}$, denoted
$C_{\mathcal{P}}(k,n)$, is the total number of bits exchanged
between the user and all servers, maximized over
$x\in\{0,1\}^{n}$, $i\in[n]$, and random string {\tt aux}.
We denote by $(k,n; C_{\mathcal{P}}(k,n))$-PIR a $k$-server
 PIR scheme with communication complexity $C_{\mathcal{P}}(k,n)$.

\citet{KatzTrevisan:2000} were the first to show generic
transformations between information-theoretic PIR schemes and LDCs.
Subsequently, \citet{Trevisan:2004}
introduced the notion of perfectly smooth decoders:

\begin{definition}[\citep{Trevisan:2004}]
A $k$-query LDC $\emph{\textbf{C}}:\Sigma^{n}\rightarrow \Gamma^{N}$
is said to have a \emph{perfectly smooth decoder} if it has a local
decoding algorithm $\mathcal{D}$ satisfying:
\begin{enumerate}
  \item In every invocation,  each query  of $\mathcal{D}$ is uniformly distributed over $[N]$.
  \item For every $x\in
  \Sigma^{n}$ and  $i\in[n]$,
  $\Pr[\mathcal{D}^{\emph{\textbf{C}}(x)}(i)=x_{i}]=1$.
\end{enumerate}
\end{definition}

LDCs with perfectly smooth decoders directly  give  information-theoretic
PIR schemes.

\begin{proposition}[\citep{Trevisan:2004}]
\label{Lem:New5}
If there is a $k$-query LDC $\textbf{{\em C}}:\Sigma^{n}\rightarrow
\Gamma^{N}$ which has a perfectly smooth decoder, then there is a
$(k,n; k(\log N+\log|\Gamma|))$\emph{-PIR}
scheme.
\end{proposition}

The LDCs obtained by \citet{Efremenko:2009} and
\citet{ItohSuzuki:2010} both have perfectly smooth decoders, and so do the LDCs
we construct in \ref{subSec:newLDC}. Applying
\ref{Lem:New5} to the Itoh-Suzuki LDCs, one obtains a family of
positive integers $\{k^{(r)}\}_{r\geq 4}$ for which $k^{(r)}\leq 3\cdot
2^{r-2}$, such that for every $r\geq 4$, there is a $k^{(r)}$-server
PIR scheme whose communication complexity is $\exp(O(\sqrt[s]{\log n(\log\log n)^{s-1}}))$,
 where $s=\log k^{(r)}+2-\log 3$. These PIR schemes are among
the most efficient PIR schemes before this work. Here, we improve their results with the following
theorem (an easy consequence of \ref{Thm:New3} and \ref{Lem:New5}).

\begin{theorem}\label{Thm:New4}
The following statements hold:
\begin{enumerate}[(a)]
\item There is a family of positive integers $\{k^{\langle r \rangle}\}_{1\leq
r\leq 103}$  for which $k^{\langle r \rangle}\leq (\sqrt{3})^{r}$ if $r$ is even, and
$k^{\langle r \rangle}\leq 8\cdot (\sqrt{3})^{r-3}$ if $r$ is odd, such that for
every $r\in[103]$, there is a $k^{\langle r \rangle}$-server PIR scheme with
communication complexity $\exp(O(\sqrt[s]{\log n(\log\log n)^{s-1}}))$,
 where $s=2\log k^{\langle r \rangle}/\log 3$ if $r$ is even, and $s=(2\log k^{\langle r \rangle}-6+3\log 3)/\log 3$
 if $r$
 is odd.

\item There is a family of positive integers $\{k^{\langle r \rangle}\}_{r\geq
104}$ for which $k^{\langle r \rangle}\leq (3/4)^{51}\cdot 2^{r}$, such that for every
$r\geq 104$ there is a $k^{\langle r \rangle}$-server PIR scheme with communication
complexity $\exp(O(\sqrt[s]{\log n(\log\log n)^{s-1}}))$,
 where $s=\log k^{\langle r \rangle}+102-51\log 3$.

\item If $|\mathbb{M}_{2,{\rm Mersenne}}|=\infty$, then there is a
family of positive integers $\{k^{\langle r \rangle}\}_{r\geq 1}$ for which
$k^{\langle r \rangle}\leq (\sqrt{3})^{r}$ if $r$ is even, and $k^{\langle r \rangle}\leq 8\cdot
(\sqrt{3})^{r-3}$ if $r$ is odd, such that for every $r\geq 1$ there
is a $k^{\langle r \rangle}$-server PIR scheme with communication complexity
$\exp(O(\sqrt[s]{\log n(\log\log n)^{s-1}}))$,
 where $s=2\log k^{\langle r \rangle}/\log 3$ if $r$ is even, and $s=(2\log k^{\langle r \rangle}-6+3\log 3)/\log 3$ if $r$
 is odd.
\end{enumerate}
\end{theorem}

\section{Conclusion}
\label{Sec:conclusion}

In this paper, we showed that every Mersenne number which is the
product of two primes can be used to improve the query complexity  by a factor of 3/4
in Efremenko's framework for constructing LDCs. Based on the 50 elements in
$\mathbb{M}_{2,{\rm Mersenne}}$ we discovered, a new family of
query-efficient LDCs of subexponential length with better
performance than those of \citet{Efremenko:2009} and \citet{ItohSuzuki:2010}
were obtained. Applying our new
LDCs to the construction of PIR schemes, we obtained a new family of PIR
schemes, which  are also more efficient than those of
\citet{Efremenko:2009} and \citet{ItohSuzuki:2010}. It is an interesting open problem to
determine whether $|\mathbb{M}_{2,{\rm Mersenne}}|=\infty$.
Furthermore, identifying new elements in $\mathbb{M}_{2,{\rm
Mersenne}}$ can improve our results and is also of interest on its
own right.

\begin{acknowledge}
The authors are grateful to Oded Goldreich for valuable suggestions that helped
improve the presentation of the paper.
The authors also thank Joachim von zur Gathen and the anonymous referee
for helpful comments.

Research of Y. M. Chee, S. Ling, and H. Wang is supported in part by the National Research
Foundation of Singapore under Research Grant NRF-CRP2-2007-03.
\end{acknowledge}

\bibliography{/Users/ymchee/Documents/Bibliography/mybibliography}

\end{document}